\documentclass[11pt]{article}

\usepackage[margin=1in]{geometry}
\usepackage{amsthm,amsfonts,amssymb,amsmath}
\usepackage{bm}
\usepackage{hyperref}
\usepackage{graphicx}
\usepackage{url}
\usepackage{hyperref} 
\hypersetup{
    colorlinks=true,
    linkcolor=blue,
    citecolor=blue,
    filecolor=blue,
    urlcolor=blue
}
\usepackage[all]{hypcap}
\usepackage[comma,square,numbers,sort]{natbib}
\mathchardef\ordinarycolon\mathcode`\:
\mathcode`\:=\string"8000
\begingroup \catcode`\:=\active
  \gdef:{\mathrel{\mathop\ordinarycolon}}
\endgroup
\usepackage{color}

\newcommand{\seg}{z}
\newcommand{\nn}{\nonumber \\}
\newcommand{\bra}[1]{\langle{#1}|}
\newcommand{\ket}[1]{|{#1}\rangle}
\newcommand{\norm}[1]{\|{#1}\|}
\newcommand{\CC}{\mathbb{C}}
\newcommand{\openone}{\mathbb{I}}
\newcommand{\id}{\openone}
\DeclareMathOperator{\spn}{span}
\newcommand{\iters}{\ell}
\newcommand{\vecU}{\vec U}
\DeclareMathOperator{\sel}{select}
\newcommand{\zer}{\varsigma}

\newtheorem{theorem}{Theorem}
\newtheorem{lemma}[theorem]{Lemma}

\newcommand{\eq}[1]{(\ref{eq:#1})}
\renewcommand{\sec}[1]{Section~\ref{sec:#1}}
\newcommand{\thm}[1]{Theorem~\ref{thm:#1}}
\newcommand{\lem}[1]{Lemma~\ref{lem:#1}}

\begin{document}

\title{Hamiltonian simulation with nearly \\ optimal dependence on all parameters}
\author{
Dominic W.\ Berry
\footnote{Department of Physics and Astronomy, Macquarie University, Sydney, Australia. 
\texttt{dominic.berry@mq.edu.au}
} \and 
Andrew M.\ Childs
\footnote{Department of Combinatorics \& Optimization and Institute for Quantum Computing, University of Waterloo, Waterloo, Ontario, Canada.}
\footnote{Department of Computer Science, Institute for Advanced Computer Studies, and Joint Center for Quantum Information and Computer Science, University of Maryland, College Park, Maryland, USA.
\texttt{amchilds@umd.edu}
} \and
Robin Kothari
\footnote{David R.~Cheriton School of Computer Science and Institute for Quantum Computing, University of Waterloo, Waterloo, Ontario, Canada.}
\footnote{Center for Theoretical Physics, Massachusetts Institute of Technology, Cambridge, Massachusetts, USA.
\texttt{rkothari@mit.edu}
}}
\date{}
\maketitle

\begin{abstract}
We present an algorithm for sparse Hamiltonian simulation whose complexity is optimal (up to log factors) as a function of all parameters of interest. Previous algorithms had optimal or near-optimal scaling  in some parameters at the cost of poor scaling in others. Hamiltonian simulation via a quantum walk has optimal dependence on the sparsity at the expense of poor scaling in the allowed error. In contrast, an approach based on fractional-query simulation provides optimal scaling in the error at the expense of poor scaling in the sparsity. Here we combine the two approaches, achieving the best features of both. By implementing a linear combination of quantum walk steps with coefficients given by Bessel functions, our algorithm's complexity (as measured by the number of queries and 2-qubit gates) is logarithmic in the inverse error, and
nearly linear in the product $\tau$ of the evolution time, the sparsity, and the magnitude of the largest entry of the Hamiltonian.  Our dependence on the error is optimal, and we prove a new lower bound showing that no algorithm can have sublinear dependence on $\tau$.
\end{abstract}

\section{Introduction}

The problem of simulating the dynamics of quantum systems was the original motivation for quantum computers \cite{Fey82} and remains one of their major potential applications. Although classical algorithms for this problem are inefficient, a significant fraction of the world's computing power today is spent in solving instances of this problem that arise in, e.g., quantum chemistry and materials science~\cite{OLCF13,NERSC13}. Furthermore, efficient classical algorithms for this problem are unlikely to exist: since the simulation problem is \textsf{BQP}-complete~\cite{Fey85}, an efficient classical algorithm for quantum simulation would imply an efficient classical algorithm for any problem with an efficient quantum algorithm (e.g., integer factorization \cite{Sho97}).

The first explicit quantum simulation algorithm, due to Lloyd \cite{Llo96}, gave a method for simulating Hamiltonians that are sums of local interaction terms. Aharonov and Ta-Shma gave an efficient simulation algorithm for the more general class of sparse Hamiltonians \cite{AT03}, and much subsequent work has given improved simulations \cite{Chi04,BACS07,WBHS11,Chi10,PQSV11,BC12,CW12,BCCKS14,BCCKS15}.  Sparse Hamiltonians include most physically realistic Hamiltonians as a special case (making these algorithms potentially useful for simulating real-world systems). In addition, sparse Hamiltonian simulation can be used to design other quantum algorithms \cite{HHL09,CCDFGS03,CCJY09}.
For example, it was used to convert the algorithm for evaluating a balanced binary NAND tree with $n$ leaves \cite{FGG08} to the discrete-query model \cite{CCJY09}.

In the Hamiltonian simulation problem, we are given an $n$-qubit Hamiltonian $H$ (a Hermitian matrix of size $2^n \times 2^n$), an evolution time $t$, and a precision $\epsilon>0$, and are asked to implement the unitary operation $e^{-iHt}$
up to error at most $\epsilon$ (as quantified by the diamond norm distance).
That is, the task is to implement a unitary operation, rather than simply to generate \cite{AMRR11} or convert \cite{LMRSS11} a quantum state.
We say that $H$ is $d$-sparse if it has at most $d$ nonzero entries in any row.  In the sparse Hamiltonian simulation problem, $H$ is specified by a black box that takes input $(j,\ell) \in [2^n] \times [d]$ (where $[d] := \{1,\ldots,d\}$) and outputs the location and value of the $\ell$th nonzero entry in the $j$th row of $H$.
Specifically, as in \cite{BC12}, we assume access to an oracle $O_H$ acting as
\begin{equation}
O_H\ket{j,k,z} = \ket{j,k,z \oplus{H_{jk}}}
\label{eq:oracleh}
\end{equation}
for $j,k \in [2^n]$ and bit strings $z$ representing entries of $H$, and another oracle $O_F$ acting as
\begin{equation}
O_F\ket{j,\ell} = \ket{j,f(j,\ell)},
\label{eq:oraclef}
\end{equation}
where $f(j,\ell)\colon [2^n] \times [d] \to [2^n]$ is a function giving the column index of the $\ell$th nonzero element in row $j$.  Note that the form of $O_F$ assumes that the locations of the nonzero entries of $H$ can be computed in place.  This is possible if we can efficiently compute both $(j,\ell) \mapsto f(j,\ell)$ and the reverse map $(j,f(j,\ell)) \mapsto \ell$, which holds in typical applications of sparse Hamiltonian simulation.  Alternatively, if $f$ provides the nonzero elements in order, we can compute the reverse map with only a $\log d$ overhead by binary search.

At present, the best algorithms for sparse Hamiltonian simulation, in terms of query complexity (i.e., the number of queries made to the oracles) and number of 2-qubit gates used, are one based on a Szegedy quantum walk \cite{Chi10,BC12} and another based on simulating an unconventional model of query complexity called the fractional-query model 
\cite{BCCKS14}.
An algorithm with similar complexity to \cite{BCCKS14} is based on implementing a Taylor series of the exponential \cite{BCCKS15}.
The quantum walk approach has query complexity $O(d\norm{H}_{\max}t/\sqrt\epsilon)$, which is linear in both the sparsity $d$ and the evolution time $t$.  (Here $\norm{H}_{\max}$ denotes the largest entry of $H$ in absolute value.)  However, this approach has poor dependence on the allowed error $\epsilon$. In contrast, the fractional-query approach has query complexity $O\big(d\tau \frac{\log(\tau/\epsilon)}{\log\log(\tau/\epsilon)}\big)$, where $\tau := d\norm{H}_{\max}t$.  This approach gives exponentially better dependence on the error at the expense of quadratically worse dependence on the sparsity.
Considering the fundamental importance of quantum simulation, it is desirable to have a method that achieves the best features of both approaches.
In this work, we combine the two approaches, giving the following.
\begin{theorem}
\label{thm:upper}
A d-sparse Hamiltonian $H$ acting on $n$ qubits can be simulated for time $t$ within error $\epsilon$ with 
\begin{equation}
\label{eq:upper}
  O\left( \tau \frac{\log(\tau/\epsilon)}{\log\log(\tau/\epsilon)}\right)
\end{equation}
queries and 
\begin{equation}
  O\left( \tau [n + \log^{5/2}(\tau/\epsilon)] \frac{\log(\tau/\epsilon)}{\log\log(\tau/\epsilon)}\right) \end{equation}
additional 2-qubit gates, where $\tau := d \norm{H}_{\max} t$.
\end{theorem}

This result provides a strict improvement over the query complexity of \cite{BCCKS14,BCCKS15}, removing a factor of $d$ in $\tau$, and thus providing near-linear instead of superquadratic dependence on $d$.

We also prove a lower bound showing that any algorithm must use $\Omega(\tau)$ queries. While a lower bound of $\Omega(t)$ was known previously \cite{BACS07}, our new lower bound shows that the complexity must be at least linear in the product of the sparsity and the evolution time.  Our proof is similar to a previous limitation on the ability of quantum computers to simulate non-sparse Hamiltonians \cite{CK10}: by replacing each edge in the graph of the Hamiltonian by a complete bipartite graph $K_{d,d}$, we effectively boost the strength of the Hamiltonian by a factor of $d$ at the cost of making the matrix less sparse by a factor of $d$.  Combining this result with the error-dependent lower bound of \cite{BCCKS14}, we find a lower bound as follows.

\begin{theorem}
\label{thm:lower}
For any $\epsilon,t>0$, integer $d\ge 2$, and fixed value of $\norm{H}_{\max}$, there exists a $d$-sparse Hamiltonian $H$ such that simulating $H$ for time $t$ with precision $\epsilon$ has query complexity
\begin{equation}
\label{eq:lower}
\Omega \left( \tau + \frac{\log(1/\epsilon)}{\log\log(1/\epsilon)}\right).
\end{equation}
\end{theorem}

Thus our result is near-optimal for the scaling in either $\tau$ or $\epsilon$ on its own. However, our upper bound \eq{upper} has a product, whereas the lower bound \eq{lower} has a sum. It remains an open question how to close the gap between these bounds. Intriguingly, a slight modification of our technique gives another algorithm with the following complexity.

\begin{theorem}
\label{thm:tradeoff}
For any $\alpha\in(0,1]$, a d-sparse Hamiltonian $H$ acting on $n$ qubits can be simulated for time $t$ within error $\epsilon$ with
query complexity
\begin{equation}
\label{eq:finalresult2} O\bigl( \tau^{1+\alpha/2} + \tau^{1-\alpha/2} \log(1/\epsilon)\bigr).
\end{equation}
\end{theorem}

This result provides a nontrivial tradeoff between the parameters $t$, $d$, and $\epsilon$, and suggests that further improvements to such tradeoffs may be possible.

We now informally describe the key idea behind our algorithms.
For simplicity, suppose that the entries of the Hamiltonian are small, satisfying $\norm{H}_{\max} \leq 1/d$, and $t=1$.
Previous work on Hamiltonian simulation \cite{Chi10,BC12} has shown that using a constant number of queries, we can construct a unitary $U$ whose top-left block (in some basis) is exactly $e^{-i\arcsin(H)}$.
Technical difficulties aside, the essential problem is to implement the unitary $e^{-iH}$ given the ability to perform $e^{-i\arcsin(H)}$.
While it is not clear how to express $e^{-iH}$ as a product of easy-to-implement unitaries and $e^{-i\arcsin(H)}$, it can be approximated by a linear combination of powers of $e^{-i\arcsin(H)}$. 
Although such a decomposition may not seem natural, we show that nevertheless it leads to an efficient implementation.

In the next section we present a more technical overview of this high-level idea. In \sec{analysis} we analyze and prove the correctness of our algorithms. \sec{lower} proves the lower bound presented in \thm{lower} and we conclude with some discussion in \sec{disc}.

\section{Overview of algorithms}

Our algorithm uses a Szegedy quantum walk as in \cite{Chi10,BC12}, but with a linear combination of different numbers of steps.  Such an operation can be implemented using the techniques that were developed to simulate the fractional-query model \cite{BCCKS14}.  This allows us to introduce a desired phase more accurately than with the phase estimation approach of \cite{Chi10,BC12}.  As in \cite{BCCKS14}, we first implement the approximated evolution for some time interval with some amplitude and then use oblivious amplitude amplification to make the implementation deterministic, facilitating simulations for longer times.
In the rest of this section, we describe the approach in more detail.

References~\cite{Chi10,BC12} define a quantum walk step $U$ that depends on the Hamiltonian $H$ to be simulated.
In turn, this quantum walk step is based on a state preparation procedure that only requires one call to the sparse Hamiltonian oracle, avoiding the need to decompose $H$ into a sum of terms as in product-formula approaches.
Two copies of the Hilbert space acted on by $H$ are used.
First, the initial state is in one of these Hilbert spaces.
Then, the state preparation procedure is used to map the initial state onto the joint Hilbert space.
This state preparation acts on the second copy of the Hilbert space, controlled by the state in the first Hilbert space.
The quantum walk steps take place in this joint Hilbert space.
Finally, the controlled state preparation is inverted to map the final state back to the first Hilbert space.
 
In the controlled state preparation, each eigenstate of $H$ is mapped onto a superposition of two eigenstates $\ket{\mu_\pm}$ of the quantum walk step $U$.
The precise definition of $U$ is not needed here; for our application, it suffices to observe that the
eigenvalues $\mu_\pm$ of $U$ are related to the eigenvalues $\lambda$ of $H$ via
\begin{equation}
\label{eq:mulam}
\mu_\pm = \pm e^{\pm i\arcsin(\lambda/Xd)},
\end{equation}
where $X\ge \norm{H}_{\max}$ is a parameter that can be increased to make the steps of the quantum walk closer to the identity.
For small $\lambda/Xd$, the steps of the quantum walk yield a phase factor that is nearly proportional to that for the Hamiltonian evolution.
However, the phase deviates from the desired value since the function $\arcsin\nu$ is not precisely linear about $\nu=0$.
Also, there are two eigenvalues $\mu_\pm$, and in previous approaches it was necessary to distinguish between these to approximate Hamiltonian evolution~\cite{Chi10,BC12}.
In contrast, for the new technique we present here it is not necessary to distinguish the eigenspaces.

An obvious way to increase the accuracy is to increase $X$ above its minimum value of $\norm{H}_{\max}$. However, the number of steps of the quantum walk is $O(tXd)$, so increasing $X$ results in a less efficient simulation. Another approach is to use phase estimation to correct the phase factor \cite{Chi10,BC12}, but this approach still gives polynomial dependence on $1/\epsilon$.

Instead, we propose using a superposition of steps of the quantum walk to effectively linearize the $\arcsin$ function. Specifically, rather than applying $U$, we apply
\begin{equation} \label{eq:super}
V_k := \sum_{m=-k}^{k} a_m U^{m}
\end{equation}
for some coefficients $a_{-k},\ldots,a_k$.
We show that the coefficients can be chosen by considering the generating function for the Bessel function \cite[9.1.41]{AS64},
\begin{equation}
\label{eq:generat}
\sum_{m=-\infty}^{\infty} J_m(\seg) \mu_\pm^{m} = \exp\left[ \frac \seg 2 \left( \mu_\pm - \frac 1{\mu_\pm} \right) \right] = e^{i\lambda \seg/Xd},
\end{equation}
where the second equality follows from \eq{mulam}.
Because the right-hand-side does not depend on whether the eigenvalue of $U$ is $\mu_+$ or $\mu_-$, there is no need to distinguish the eigenspaces.
Thus the ability to perform the operation
\begin{equation}
  \label{eq:infinite}
  \sum_{m=-\infty}^\infty J_m(\seg) U^{m}
\end{equation}
would allow us to exactly implement the evolution under $H$ for time $-\seg/Xd$.
Because of the minus sign, we will take $\seg$ to be negative to obtain positive time.
By truncating the sum in \eq{infinite} to some finite range $\{-k,\ldots,k\}$, we obtain an expression in which each term can be performed using at most $k$ queries.  Because the Bessel function falls off exponentially for large $|m|$, we can obtain error at most $\epsilon$ with a cutoff $k$ that is only logarithmic in $1/\epsilon$.

A linear combination of unitaries (LCU) such as \eq{super} can be implemented using the LCU Lemma (\lem{approxV}) described in the next section.
The high-level intuition for the procedure is as follows.
We prepare ancilla qubits in a superposition encoding the coefficients of the linear combination and then perform the unitary operations of the linear combination in superposition, controlled by the ancilla.
One could then obtain $V_k$ by postselecting on an appropriate ancilla state.
Instead, to obtain $V_k$ deterministically, we apply the oblivious amplitude amplification procedure introduced in \cite{BCCKS14}.
Rather than using $V_k$ to implement evolution over the entire time, we break the time up into shorter time steps we call ``segments'' (named by analogy to the segments used in \cite{BCCKS14}) and use $V_k$ to achieve the time evolution for each segment.

The complexity of our algorithm is the number of segments ($tXd/|\seg|$) times the complexity for each segment ($k$) times the number of steps needed for oblivious amplitude amplification ($a$).
We have some freedom in choosing $\seg$, which controls the amount of evolution time simulated by each segment.
To obtain near-linear dependence on the evolution time $t$, we choose $\seg=O(1)$.
Then amplitude amplification requires $O(1)$ steps, and the number of segments needed is $O(\tau)$, giving the linear factor in \eq{upper}.
The value of $k$ needed to achieve overall error at most $\epsilon$ is logarithmic in $\tau/\epsilon$, yielding the logarithmic factor in \eq{upper}.

An alternative approach is to use a larger segment that scales with $\tau$.
Choosing $\seg=-\tau^\alpha$ for $\alpha\in(0,1]$, we need $k = O(\tau^\alpha+\log(1/\epsilon))$.
Then we require $O(\tau^{1-\alpha})$ segments and $O(\tau^{\alpha/2})$ steps of amplitude amplification, giving the scaling presented in \thm{tradeoff}.

\section{Analysis of algorithms}
\label{sec:analysis}
\subsection{A quantum walk for any Hamiltonian}

We begin by reviewing the quantum walk defined in \cite{Chi10,BCCKS14}.  Given a Hamiltonian $H$ acting on $\CC^N$ (where $N:=2^n$), the Hilbert space is expanded to $\CC^{2N}\otimes\CC^{2N}$.
First, an ancilla qubit in the state $\ket{0}$ is appended, which expands the space from $\CC^N$ to $\CC^{2N}$.
Then the entire Hilbert space is duplicated, giving $\CC^{2N}\otimes\CC^{2N}$.
This is achieved using the isometry
\begin{equation}
T:=\sum_{j=0}^{N-1} \sum_{b\in\{0,1\}} (\ket{j}\bra{j}\otimes \ket{b}\bra{b}) \otimes \ket{\varphi_{jb}}
\end{equation}
with $\ket{\varphi_{j1}}=\ket{0}\ket{1}$ and
\begin{equation}
\label{eq:state}
\ket{\varphi_{j0}} := \frac{1}{\sqrt d} \sum_{\ell\in F_j} \ket{\ell} \Biggl( \sqrt{\frac{H^*_{j\ell}}{X}}\ket{0}+\sqrt{1-\frac{|H^*_{j\ell}|}{X}}\ket{1}\Biggr),
\end{equation}
where $X\ge \norm{H}_{\max}$ and $F_j$ is the set of indices of nonzero elements in column $j$ of $H$.
Here we use the convention that the first subsystem is the original space, the next is the ancilla qubit, and the third and fourth subsystems are the duplicated space and duplicated ancilla qubit, respectively.
This operation can be viewed as a controlled state preparation, creating state $\ket{\varphi_{j0}}$ on input $\ket{j}\ket{0}$.
If the ancilla qubit is in the state $\ket{1}$, then $\ket{0}\ket{1}$ is prepared.
Starting with the initial space, the controlled state preparation is performed, and then steps of the quantum walk are applied using the unitary
\begin{equation}
U := iS(2TT^\dagger -\openone),
\end{equation}
where $S$ swaps the two registers (i.e., $S\ket{j_1}\ket{j_2}\ket{\ell_1}\ket{\ell_2}=\ket{\ell_1}\ket{\ell_2}\ket{j_1}\ket{j_2}$ for all $j_1,\ell_1 \in [N]$, $j_2,\ell_2\in\{0,1\}$).
Finally, the inverse state preparation $T^\dagger$ is performed.
For a successful simulation, the output should lie in the original space, and the ancilla should be returned to the state $\ket{0}$.

Let $\lambda$ be the eigenvalue of $H$ with eigenstate $\ket{\lambda}$, and let $\nu:=\lambda/Xd$ be the corresponding scaled eigenvalue for the quantum walk. The steps of the quantum walk $U$ satisfy $U\ket{\mu_\pm} = \mu_\pm\ket{\mu_\pm}$ \cite{Chi10} with
\begin{align}
\ket{\mu_\pm} &:= (T+ i\mu_\pm ST)\ket{\lambda},\\
\label{eq:munueq}
\mu_\pm &:= \pm\sqrt{1-\nu^2} + i\nu = \pm e^{\pm i\arcsin\nu}.
\end{align}
To apply the steps of the quantum walk to approximate Hamiltonian evolution, there are two challenges: we must handle both the $\ket{\mu_+}$ and $\ket{\mu_-}$ sectors, and correct the applied phase.
In this work we are able to solve both these challenges at once by using a superposition of steps of the quantum walk.

\subsection{Linear combination of unitaries}

We now describe how to perform a linear combination of unitary operations.
Given an $M$-tuple of unitary operations $\vecU = (U_1,\ldots,U_M)$, we quantify the complexity of implementing a linear combination of the $U_m$s in terms of the number of invocations of
\begin{equation}
\label{eq:controlled}
\sel(\vecU) := \sum_{m=1}^M \ket{m}\bra{m} \otimes U_{m}.
\end{equation}
Such a result was previously given in \cite{BCCKS15,Kot14}.
Here we formalize that result and generalize to allow more steps of oblivious amplitude amplification.
The overall result is as given in the following lemma.

\begin{lemma}[LCU Lemma]\label{lem:approxV}
Let $\vecU = (U_1,\ldots,U_M)$ be unitary operations and let $\tilde V = \sum_{m=1}^M a_m U_m$ be $\delta$-close to a unitary.
We can approximate $\tilde V$ to within $O(\delta)$ using $O(a)$ $\sel(\vecU)$ and $\sel(\vecU^\dag)$ operations and $O(Ma)$ additional 2-qubit gates, where $a:=\sum_{m=1}^M |a_m|$.
\end{lemma}

To prove this result, we first consider an operation that would give $\tilde V$ with postselection, then apply oblivious amplitude amplification to achieve it deterministically.
The operation that provides $\tilde V$ with postselection is described in the following lemma.

\begin{lemma}\label{lem:sup}
Let $\vecU = (U_1,\ldots,U_M)$ be unitary operations acting on a Hilbert space $\mathcal{H}_2$, let $\tilde V = \sum_{m=1}^M a_m U_m$,
and let $s\ge\sum_{m=1}^M |a_m|$ be a real number.
Define a Hilbert space $\mathcal{H}_1$ to be a tensor product of a qubit and a subspace of dimension $M$, and let $\ket{\zer} := \ket{0}\ket{0} \in\mathcal{H}_1$.
Then there exists a unitary operation $W$ acting on $\mathcal{H}_1 \otimes \mathcal{H}_2$ such that 
$Z = \frac 1s (\ket{\zer}\bra{\zer}\otimes \tilde{V})$, with $Z:=PWP$, $P := \ket{\zer}\bra{\zer} \otimes \id$.
The operation $W$ can be applied with $O(1)$ $\sel(\vecU)$ and $\sel(\vecU^\dag)$ operations and $O(M)$ additional 2-qubit gates.
\end{lemma}

\begin{proof}
To perform $W$, we perform an operation that rotates the ancilla qubits from $\ket\zer$ to the state
\begin{equation}
\ket{\chi} = \left( \sqrt{\frac as}\ket 0+\sqrt{1-\frac as}\ket 1 \right) \otimes \frac{1}{\sqrt{a}} \sum_{m=1}^M \sqrt{a_m}\ket{m},
\end{equation}
where $a:=\sum_{m=1}^M |a_m|$.
This state is of dimension $2M$ and can be prepared from state $\ket{\zer}$ using $O(M)$ operations (which is trivial for $\ket{m}$ encoded in unary).
Next we perform the controlled operation $\sel(\vecU)$.
Finally, inverting the preparation of $\ket{\chi}$ and projecting onto $\ket{\zer}$ would effectively project the ancilla onto $\ket{\chi}$.
Then the unnormalized operation on $\mathcal{H}_2$ is $\tilde V/s$, corresponding to $Z$.
The action of applying the unitary operation to prepare $\ket{\chi}$, the controlled operation $\sel(\vecU)$, and the inverse preparation
gives the desired operation $W$.
\end{proof}

Next we provide a multi-step version of robust amplitude amplification, generalizing the single-step version presented in \cite{BCCKS15}.
In this lemma, and throughout this paper, $\norm{\cdot}$ denotes the spectral norm.

\begin{lemma}[Robust oblivious amplitude amplification]\label{lem:roaa}
Let $W$ be a unitary matrix acting on $\mathcal{H}_1 \otimes \mathcal{H}_2$ and let $P$ be the projector onto the subspace whose first register is $\ket{\zer} := \ket{0}\ket{0} \in \mathcal{H}_1$, i.e., $P := \ket{\zer}\bra{\zer} \otimes \id$. Furthermore let  $Z:=PWP$ satisfy $Z = \frac 1s (\ket{\zer}\bra{\zer}\otimes \tilde{V})$, where $\tilde{V}$ is $\delta$-close to a unitary matrix and $\sin\bigl(\frac{\pi}{2(2\iters+1)}\bigr)= \frac 1s$ for some $\iters \in \mathbb{N}$, and let $R:=-W(\id-2P)W^\dag(\id-2P)$. Then 
\begin{equation}
\label{eq:roaa}
\norm{PR^\iters WP - (\ket{\zer}\bra{\zer} \otimes \tilde{V})} = O(\delta).
\end{equation}
\end{lemma}

\begin{proof}
We start by considering a single iteration, as in \cite{BCCKS15}. Then we have
\begin{align}
RWP 
&= -W(\openone-2P)W^\dagger(\openone-2P)WP \nn
&= -WP + 2WP + 2PWP -4WPW^\dagger PWP \nn
&= WP+2 Z -4W Z^\dagger  Z.
\end{align}
Multiplying by $P$ on the left gives
\begin{align}
PRWP = -PW(\openone-2P)W^\dagger(\openone-2P)WP = 3 Z -4 Z  Z^\dagger  Z,
\end{align}
which matches the expression in \cite{BCCKS15}. 

The general solution after $m$ iterations is
\begin{align}
\label{eq:sol}
R^m WP = (WP- Z)\frac{T_{2m+1}(\sqrt{1- Z^\dagger  Z})}{\sqrt{1- Z^\dagger  Z}} +  Z\frac{(-1)^m T_{2m+1}(\sqrt{ Z^\dagger  Z})}{\sqrt{ Z^\dagger  Z}},
\end{align}
where $T_{2m+1}$ are Chebyshev polynomials of the first kind.
(Because Chebyshev polynomials for odd order only include odd powers, no square roots appear when \eq{sol} is expanded.)

We establish \eq{sol} by induction. First note that it holds for $m=0$, because $T_1(x)=x$, so the right-hand side evaluates to $WP$.
Next assume that it holds for a given $m$.
It is straightforward to show that
\begin{align}
R(WP-Z)  &= (WP- Z)(1-2 Z^\dagger  Z) +2 Z(1- Z^\dagger  Z) \quad \text{and} \nn
 R  Z  &=  (WP- Z)(-2 Z^\dagger  Z) +  Z(1-2 Z^\dagger  Z).
\end{align}
Hence, multiplying both sides of \eqref{eq:sol} by $R$, we get
\begin{align}
\label{eq:hard}
R^{m+1}WP &=  R \left[(WP- Z)\frac{T_{2m+1}(\sqrt{1- Z^\dagger  Z})}{\sqrt{1- Z^\dagger  Z}} +  Z\frac{(-1)^m T_{2m+1}(\sqrt{ Z^\dagger  Z})}{\sqrt{ Z^\dagger  Z}}\right] \nn
&= (WP- Z)\left[(1-2 Z^\dagger  Z)\frac{T_{2m+1}(\sqrt{1- Z^\dagger  Z})}{\sqrt{1- Z^\dagger  Z}} -2 Z^\dagger  Z\frac{(-1)^m T_{2m+1}(\sqrt{ Z^\dagger  Z})}{\sqrt{ Z^\dagger  Z}}\right] \nn
&\quad +  Z\left[ 2(1- Z^\dagger  Z)\frac{T_{2m+1}(\sqrt{1- Z^\dagger  Z})}{\sqrt{1- Z^\dagger  Z}} + (1-2 Z^\dagger  Z)\frac{(-1)^m T_{2m+1}(\sqrt{ Z^\dagger  Z})}{\sqrt{ Z^\dagger  Z}} \right].
\end{align}
To progress further, we use the relation \cite[22.3.15]{AS64}
\begin{equation}
T_{2m+1}(x) = \cos[(2m+1) \arccos x] = (-1)^m\sin[(2m+1)\arcsin x].
\end{equation}
Using this we find that, with $x=\sin\theta$,
\begin{align}
&(1-2x^2)\frac{T_{2m+1}(\sqrt{1-x^2})}{\sqrt{1-x^2}}-2 x^2\frac{(-1)^m T_{2m+1}(x)}{x} \nn
&\quad=(\cos^2\theta-\sin^2\theta)\frac{\cos[(2m+1)\theta]}{\cos\theta}-2 \sin\theta \sin[(2m+1)\theta] \nn
&\quad=\frac{\cos (2\theta)\cos[(2m+1)\theta]-\sin(2\theta) \sin[(2m+1)\theta]}{\cos\theta} \nn
&\quad=\frac{\cos[(2m+3)\theta]}{\cos\theta} \nn
&\quad=\frac{T_{2m+3}(\sqrt{1-x^2})}{\sqrt{1-x^2}}.
\end{align}
Next put $x=\cos\phi$ to obtain
\begin{align}
&2(1-x^2)\frac{T_{2m+1}(\sqrt{1-x^2})}{\sqrt{1-x^2}}+(1-2 x^2)\frac{(-1)^m T_{2m+1}(x)}{x} \nn
&\quad=2(\sin^2\phi)\frac{(-1)^m\sin[(2m+1)\phi]}{\sin\phi} -(\cos^2\phi-\sin^2\phi)\frac{(-1)^m \cos[(2m+1)\phi]}{\cos\phi} \nn
&\quad=(-1)^m\frac{\sin(2\phi)\sin[(2m+1)\phi]-\cos(2\phi) \cos[(2m+1)\phi]}{\cos\phi} \nn
&\quad=(-1)^{m+1}\frac{\cos[(2m+3)\phi]}{\cos\phi} \nn
&\quad=(-1)^{m+1}\frac{T_{2m+3}(x)}{x}.
\end{align}
Using these relations, \eq{hard} simplifies to
\begin{align}
R^{m+1}WP = 
(WP- Z)\frac{T_{2m+3}(\sqrt{1- Z^\dagger  Z})}{\sqrt{1- Z^\dagger  Z}}+  Z\frac{(-1)^{m+1} T_{2m+3}(\sqrt{ Z^\dagger  Z})}{\sqrt{ Z^\dagger  Z}}.
\end{align}
Hence we find that, if \eq{sol} is correct for non-negative integer $m$, it is correct for $m+1$.
Hence it is correct for all non-negative integers $m$ by induction.

Thus we find that by multiplying on the left by $P$, we obtain
\begin{align}
\label{eq:sol2}
PR^m WP =  Z\frac{(-1)^m T_{2m+1}(\sqrt{ Z^\dagger  Z})}{\sqrt{ Z^\dagger  Z}}.
\end{align}
Then, in the case that $\delta=0$, i.e., if $\tilde{V}$ is equal to a unitary $V$, we would have
\begin{equation}
 Z=\frac 1s (\ket{\zer}\bra{\zer}\otimes V).
\end{equation}
Then $Z^\dagger  Z=(\ket{\zer}\bra{\zer}\otimes \openone)/s^2$, and we get
\begin{align}
 Z\frac{(-1)^m T_{2m+1}(\sqrt{ Z^\dagger  Z})}{\sqrt{ Z^\dagger  Z}} &= \frac 1s (\ket{\zer}\bra{\zer}\otimes V) \frac{(-1)^m T_{2m+1}(1/s)}{1/s} \nn
&= (\ket{\zer}\bra{\zer}\otimes V)(-1)^m T_{2m+1}(1/s) \nn
&= (\ket{\zer}\bra{\zer}\otimes V)\sin[(2m+1)\arcsin(1/s)].
\end{align}
In the case $m=\iters$, we can use $\sin\bigl(\frac{\pi}{2(2\iters+1)}\bigr)= \frac 1s$ to obtain
$\sin[(2\iters+1)\arcsin(1/s)]=1$, which implies
\begin{equation}
PR^\iters WP = \ket{\zer}\bra{\zer}\otimes V.
\end{equation}

Next, consider the case where $\tilde{V}$ is only $\delta$-close to being unitary. Let us define
\begin{equation}
\Delta := \sqrt{\tilde V^\dagger \tilde V} - \openone.
\end{equation}
We immediately obtain $\norm{\Delta}=O(\delta)$, and
\begin{equation}
\frac 1s(\ket{\zer}\bra{\zer}\otimes\Delta) = \sqrt{ Z^\dagger  Z} - \frac 1s( \ket{\zer}\bra{\zer}\otimes\openone).
\end{equation}
We then get
\begin{align}
 Z\frac{(-1)^\iters T_{2\iters+1}(\sqrt{ Z^\dagger  Z})}{\sqrt{ Z^\dagger  Z}} &= \frac 1s (\ket{\zer}\bra{\zer}\otimes \tilde V) \frac{(-1)^\iters T_{2\iters+1}((\openone+\Delta)/s)}{(\openone+\Delta)/s} \nn
&=(\ket{\zer}\bra{\zer}\otimes \tilde V) \frac{(-1)^\iters T_{2\iters+1}((\openone+\Delta)/s)}{\openone+\Delta}.
\end{align}
Using $(-1)^\iters T_{2\iters+1}(x)=\sin[(2\iters+1)\arcsin x]$ and $(-1)^\iters T_{2\iters+1}(1/s)=1$, we obtain
\begin{equation}
\| (-1)^\iters T_{2\iters+1}((\openone+\Delta)/s) - \openone \| = O(\iters^2\delta^2/s^2).
\end{equation}
Since $\iters = \Theta(s)$, $\iters^2/s^2=O(1)$, which implies
\begin{equation}
\| (-1)^\iters T_{2\iters+1}((\openone+\Delta)/s) - \openone \| = O(\delta^2).
\end{equation}
The contribution to the error from $\openone+\Delta$ is $O(\delta)$, so we have
\begin{align}
\left\| Z\frac{(-1)^\iters T_{2\iters+1}(\sqrt{ Z^\dagger  Z})}{\sqrt{ Z^\dagger  Z}} - (\ket{\zer}\bra{\zer}\otimes \tilde V) \right\| = O(\delta).
\end{align}
Using \eq{sol2} we then get \eq{roaa} as required.
\end{proof}

\lem{roaa} is in terms of the spectral norm distance, but the diamond norm distance is at most a constant factor larger.
The specific result, proven in the Appendix, is as follows.
\begin{lemma}
\label{lem:diamond}
Let $U$, $V$ be operators satisfying $\|U\|\le 1$ and $\|V\|\le 1$, and let $\|\cdot\|_\diamond$ denote the diamond norm.  Then $\norm{U-V}_\diamond \le 2 \norm{U-V}$.  Furthermore, if $\mathcal V$ is a quantum channel with Kraus decomposition $\mathcal{V}(\rho) = V \rho V^\dag + \sum_j V_j \rho V_j^\dag$ and $\mathcal{U}(\rho) = U \rho U^\dag$, then $\norm{\mathcal{U}-\mathcal{V}}_\diamond \le 4 \norm{U-V}$.
\end{lemma}

The LCU Lemma follows by combining \lem{sup} and \lem{roaa}.

\begin{proof}[Proof of \lem{approxV}]
Using \lem{sup} we can implement the operation $W$ required for \lem{roaa} using $O(1)$ $\sel(\vecU)$ and $\sel(\vecU^\dag)$ operations and $O(M)$ additional 2-qubit gates.
We can choose $s\ge a$ such that $\sin\bigl(\frac{\pi}{2(2\iters+1)}\bigr)= \frac 1s$ for some $\iters \in \mathbb{N}$.
Then \lem{roaa} shows that $\tilde V$ can be approximated to within $O(\delta)$ using $O(\iters)$ applications of $W$ and the projection $P$.
Since $\iters=O(s)=O(a)$, the total number of $\sel(\vecU)$ and $\sel(\vecU^\dag)$ operations is $O(a)$ and the number of additional 2-qubit gates is $O(Ma)$.
\end{proof}

\subsection{Main algorithm}

The main problem with applying the quantum walk as presented in \cite{Chi10,BC12} is that $\arcsin \nu$ is a nonlinear function of $\nu$, so an imprecise phase is introduced.
To solve this, we use a superposition of different numbers of applications of $U$.
Define $V_k$ as in \eq{super}, where the choice of $\{a_m\}_{m=-k}^k$ is considered below.
The eigenvalues of $V_k$ corresponding to the eigenvalues $\mu_\pm$ of $U$ are
\begin{equation}
\mu_{\pm,k} := \sum_{m=-k}^k a_m \mu_\pm^{m}.
\end{equation}
In general $\mu_{\pm,k}$ can depend on $\pm$.
However, we will choose $a_m$ satisfying $a_{-m}=(-1)^m a_m$, which yields $\mu_{\pm,k}$ independent of $\pm$.

To see how to choose the coefficients $a_m$,
solve \eq{munueq} for $\nu$ to give
\begin{equation}
\nu = -\frac{i}{2} \left( \mu_\pm - \frac 1{\mu_\pm} \right).
\end{equation}
This implies that, for any $z$,
\begin{equation}
e^{i\nu \seg} = \exp\left[ \frac \seg 2 \left( \mu_\pm - \frac 1{\mu_\pm} \right) \right].
\end{equation}
This corresponds to the standard generating function for the Bessel function \cite[9.1.41]{AS64}, so
\begin{equation}
\label{eq:lamsum}
e^{i\nu \seg} = \exp\left[ \frac \seg 2 \left( \mu_\pm - \frac 1{\mu_\pm} \right) \right] = \sum_{m=-\infty}^{\infty} J_m(\seg) \mu_\pm^{m}.
\end{equation}
Thus we can take $a_m \approx J_m(\seg)$.
Because there are efficient classical algorithms to calculate Bessel functions, the circuit to prepare $\ket{\chi_k}$ can be designed efficiently.
Note that for large $m$, we have $|J_m(\seg)| \sim \frac{1}{m!}|\seg/2|^m$ \cite[9.3.1]{AS64}, so the values of $a_m$ are similar to the coefficients in the expansion of the exponential function. Thus the segments used here are analogous to the segments used in \cite{BCCKS15}.

To determine the complexity of this approach, we primarily need to bound the error in approximating $e^{i\nu \seg}$.  To optimize the result, we use the coefficients
\begin{equation}
\label{eq:avals}
a_m := \frac{J_m(\seg)}{\sum_{j=-k}^k J_j(\seg)}.
\end{equation}
We make this choice because the most accurate results are obtained when the values of $a_m$ sum to $1$.
Note also that this yields the symmetry $a_{-m}=(-1)^m a_m$, because $J_{-m}(\seg) = (-1)^m J_m(\seg)$ \cite[9.1.5]{AS64}.
The sum of $J_m(\seg)$ over all integers $m$ is equal to $1$ (which can be shown by putting $t=1$ in \cite[9.1.41]{AS64}), but because $k$ is finite, we normalize the values as in \eq{avals}.
With this choice, we have the following error bound, proved in the Appendix.

\begin{lemma}
\label{lem:errorbound}
With the values $a_m$ as in \eq{avals}, for $|\seg|\le k$ we have
\begin{equation}
\label{eq:error}
\| V_k - V_\infty \| = 
O\left( \frac{\norm{H}}{Xd} \frac{(\seg/2)^{k+1}}{k!}  \right).
\end{equation}
\end{lemma}

Note that $V_\infty$ is the exact unitary operation desired.
We now determine the query complexity of this approach. In fact, we prove a result that is slightly tighter than the query complexity stated in \thm{upper}.

\begin{lemma}
\label{lem:upperquery}
A d-sparse Hamiltonian $H$ acting on $n$ qubits can be simulated for time $t$ within error $\epsilon$ with complexity (quantified by the number of 2-qubit operations and controlled-$U$ and controlled-$U^\dagger$ operations)
\begin{equation}
  \label{eq:finalresult}
  O\left( \tau \frac{\log(\|H\|t/\epsilon)}{\log\log(\|H\|t/\epsilon)}\right). 
\end{equation}
\end{lemma}

\begin{proof}
Our main goal is to determine the value of $k$ needed to bound the error by $\epsilon$.
This depends on the length of time for the segments, which we can adjust by choosing the value of $\seg$.
We wish to perform each step deterministically with one step of oblivious amplitude amplification, so we should have $s=2$ in \lem{roaa}.
Using the values of $a_m$ given in \eq{avals}, this means that we should take $\seg=O(1)$, and for concreteness $\seg=-1/2$ yields $a<2$, so we can take $s=2$.
Then, using \lem{approxV} with $U_m=U^m$ and $\tilde V=V_k$, we can approximate the operation $V_k$ to within $O(\delta)$.

Given an allowable error in a segment of $\delta>0$, let us take
\begin{equation}
k=O\left( \frac{\log(\norm{H}/Xd\delta)}{\log\log(\norm{H}/Xd\delta)} \right).
\end{equation}
Then, using \lem{errorbound} and the inequality $k!>(k/e)^k$, it is straightforward to show that the error in each segment is no more than $\delta$.
Using \lem{diamond}, this bound on the error in terms of the spectral norm distance implies a bound on the diamond norm distance that is at most a constant factor larger.
For the total error to be no more than $\epsilon$, the value of $\delta$ can be no more than $\epsilon$ divided by the number of segments.
The number of segments is $O(tXd)$, which gives
\begin{equation}
\label{eq:kval}
k=O\left( \frac{\log(\norm{H}t/\epsilon)}{\log\log(\norm{H}t/\epsilon)} \right).
\end{equation}

Using \lem{approxV}, the complexity of each segment is $O(k)$ since a $\sel(\vecU)$ operation can be implemented with complexity $O(M)$, and $M=2k+1$.
It is straightforward to apply $\sel(\vecU)$ using $O(k)$ controlled-$U$ and controlled-$U^\dagger$ operations.
If $\ket m$ is encoded in unary, then each controlled operation may be just controlled on one qubit of $\ket{m}$.

The number of segments required is $O(tXd)$.
It is most efficient to take the minimum value of $X$, which is $\norm{H}_{\max}$.
Because each segment uses $O(k)$ controlled-$U$ and controlled-$U^\dagger$ operations, as well as $O(M)=O(k)$ additional 2-qubit gates,
the complexity for the simulation over time $t$ is $O(\tau k)$, 
Using the value of $k$ from \eq{kval} gives the overall complexity stated in \eq{finalresult}.
\end{proof}

Next we determine the gate complexity of this approach.
Again we give a slightly tighter result than presented in \thm{upper}.

\begin{lemma}
\label{lem:uppergate}
A d-sparse Hamiltonian $H$ acting on $n$ qubits can be simulated for time $t$ within error $\epsilon$ using
\begin{equation}
  \label{eq:finalresultg}
  O\left( \tau [n+F(\log(\|H\|t/\epsilon))] \frac{\log(\|H\|t/\epsilon)}{\log\log(\|H\|t/\epsilon)}\right)
\end{equation}
2-qubit gates, where $F(m)$ is the complexity of performing elementary functions with $m$ bits.
\end{lemma}

\begin{proof}
To obtain the gate complexity, we need to consider the procedure for performing the step $U$ in detail.
We can perform $T$ by first performing $\log d$ Hadamard gates to prepare the superposition state
\begin{equation}
\frac 1{\sqrt{d}}\sum_{\ell=0}^{d-1} \ket{\ell}\ket{0}.
\end{equation}
Here we take $d$ to be a power of $2$ without loss of generality.
(The value of $d$ can always be rounded up to the nearest power of two.)
Then the oracle $O_F$ (from \eq{oraclef}) can be used to produce the state
\begin{equation}
\frac 1{\sqrt{d}}\sum_{\ell\in F_j} \ket{\ell}\ket{0}.
\end{equation}
A call to the oracle $O_H$ for the value of an element of the Hamiltonian (from \eq{oracleh}) then gives the value of $H_{j\ell}$ in an ancilla space.
Another ancilla qubit is rotated from $\ket{0}$ to
\begin{equation}
\label{eq:rotstate}
\sqrt{\frac{H^*_{j\ell}}{X}}\ket{0}+\sqrt{1-\frac{|H^*_{j\ell}|}{X}}\ket{1}
\end{equation}
based on the value of $H_{j\ell}$.
Then inverting the oracle $O_H$ erases the value of $H_{j\ell}$ from the ancilla space.
Note that there is a sign ambiguity for the square root when $H_{j\ell}$ takes negative real values.
This is addressed in \cite{BC12} and does not affect the complexity.

To perform the step $U$, we also require the swap operation $S$, which has complexity $O(n)$ due to the number of qubits. The gate complexity is $O(n)$ from $S$, plus $O(\log d)=O(n)$ from the Hadamard gates, plus the complexity of performing the rotations to obtain the state \eq{rotstate}.
The complexity of the rotations depends on the number of bits of precision used for the entries of $H$. To obtain overall error $O(\epsilon)$, the number of bits must be $\log(\|H\|t/\epsilon)$.
To determine the rotations needed, we must also compute a square root and trigonometric functions on the output of the oracle.
If these functions can be computed with complexity $F(m)$ for $m$-bit inputs, the contribution to the overall complexity is $F(\log(\|H\|t/\epsilon))$. 

In \lem{upperquery} the complexity is quantified in terms of the number of controlled-$U$ and controlled-$U^\dagger$ operations, so to obtain the overall gate complexity we just need to multiply that complexity by the cost of $U$.
There is also a cost in terms of additional 2-qubit gates in \lem{upperquery}, but that is smaller than the gate cost of performing $U$.
Therefore, the gate complexity is equal to the complexity from \lem{upperquery} times $O(n+F(\log(\|H\|t/\epsilon)))$, which gives a gate complexity as in \eq{finalresultg}.
\end{proof}

This result depends on the complexity of elementary functions, $F(m)$, needed to calculate the rotations.
Using advanced techniques, $F(m)$ may be made close to linear in $m$ \cite{RPB}, though such advanced techniques only give an improvement for extremely high precision.
Using simple techniques based on Taylor series and long multiplication, $F(m)=O(m^{5/2})$.

The classical complexity of determining the coefficients $\{a_m\}_{m=-k}^k$ is also potentially significant.
A set of values of the Bessel function can be efficiently computed using Miller's recurrence algorithm \cite{Mil52,Olv64}.
The complexity scales as $k$ (the number of entries) times $\log(\|H\|t/\epsilon)$ (the bits of precision needed for each $J_m(\seg)$).
This is no larger than the quantum gate complexity.

Note that the gate complexity in \lem{uppergate} depends linearly on $n$, whereas the query complexity in \thm{upper} does not.
This is because performing an operation on a target state with $n$ qubits must require at least $\Omega(n)$ gates.
In contrast, the number of queries need not scale with $n$, because the queries are used to determine which gates to perform.
There is an implicit complexity of $\Omega(n)$ for the queries, because the input to a query is of size at least $n$.

The proof of \thm{upper} then follows immediately.
\begin{proof}[Proof of \thm{upper}]
 The implementation of $U$ uses $O(1)$ oracle calls, which means that the query complexity is the same as the number of controlled applications of $U$.
Noting that $\|H\|\le d\|H\|_{\max}$, \lem{upperquery} implies the query complexity in \thm{upper}, and \lem{uppergate} with $F(m)=O(m^{5/2})$ implies the gate complexity in \thm{upper}.
\end{proof}

\subsection{A tradeoff between $\tau$ and $\epsilon$}

The alternative algorithm characterized by \thm{tradeoff} uses larger segments with $\seg\propto -\tau^\alpha$ for $\alpha \in (0,1]$.
The case $\alpha=0$ corresponds to the case considered above, whereas $\alpha=1$ corresponds to a \emph{single} segment.
The analysis of this section assumes $\alpha>0$.

To analyze this algorithm, we first need to bound the absolute sum of Bessel functions.

\begin{lemma}\label{lem:sqrtsum}
The quantity
\begin{equation}\label{sdef}
{\cal S}(\seg) := \sum_{m=-\infty}^{\infty} |J_m(\seg)|
\end{equation}
is $O(\sqrt{|\seg|})$.
\end{lemma}

We prove this in the Appendix.
Using the robust version of amplitude amplification given in \lem{roaa}, we obtain the bound in \thm{tradeoff}.

\begin{proof}[Proof of \thm{tradeoff}]
Using \lem{errorbound}, Stirling's formula, and the fact that $\norm{H} \le d\norm{H}_{\max} \le Xd$, we find that for $|\seg|\le k$ the error is bounded as
\begin{equation}
\|V_k - V_\infty\| = O\bigl( (e/k)^k (\seg/2)^{k+1} \bigr).
\end{equation}
By \lem{roaa}, the error in a segment after amplitude amplification is of the same order.
Therefore, to ensure that the error in a segment is at most $\delta$, it suffices to take
\begin{equation}
k = O(|\seg| + \log(1/\delta)) = O(\tau^\alpha+\log(1/\delta)).
\end{equation}
With this value of $k$, we have $\sum_{m=-k}^k J_m(\seg) = 1 + O(\delta)$, so \lem{sqrtsum} gives
\begin{equation}
\sum_{m=-k}^k |a_m| = O({\cal S}(\seg)) = O(\sqrt{|\seg|}).
\end{equation}
This corresponds to the number of steps of oblivious amplitude amplification.
The overall complexity is therefore $O(k\sqrt{|\seg|})=O(k\tau^{\alpha/2})$ for a single segment.

The number of segments is $\tau/|\seg| \propto \tau^{1-\alpha}$.
This means that the complexity is $O(k\tau^{1-\alpha/2})$.
The value of $k$ also depends on the number of segments.
We can take $\delta=\epsilon/\tau^{1-\alpha}$, which gives $k=O(\tau^\alpha+\log(1/\epsilon))$, implying the result in
\thm{tradeoff}.
\end{proof}

In this proof we have ignored $\log\tau$ in comparison to $\tau^\alpha$, which would not be valid for $\alpha=0$.
For the gate complexity, we again have a multiplying factor of $n+F(\log(\|H\|t/\epsilon))$, yielding a number of gates scaling as
\begin{equation}
O\bigl( [n+F(\log(\|H\|t/\epsilon))] \bigl[\tau^{1+\alpha/2} + \tau^{1-\alpha/2} \log(1/\epsilon)\bigr]\bigr).
\end{equation}

\section{Lower bound}
\label{sec:lower}

We now present a lower bound showing that the dependence of our algorithm on $\tau := d \norm{H}_{\max} t$ is nearly optimal (and that the dependence of \cite{BC12} on $\tau$ is optimal).  The main idea of the proof is the same as in Theorem 3 of \cite{CK10}, but we slightly adapt that argument to let $t$ vary independently of $d$. Note that this is stronger than proving separate lower bounds of $\Omega(t)$ and $\Omega(d)$, since that would only show a lower bound of $\Omega(d+t)$, which is weaker than our $\Omega(td)$ lower bound.

\begin{lemma}
\label{lem:low}
For any positive integer $d$ and any $t>0$, there exists a $2d$-sparse Hamiltonian $H$ with $\norm{H}_{\max}=\Theta(1)$ such that simulating $H$ with constant precision for time $t$ requires $\Omega(td)$ queries.
\end{lemma}

\begin{proof}
Similarly to the $\Omega(t)$ lower bound from \cite{BACS07}, we construct a sparse Hamiltonian whose dynamics compute the parity of a bit string, and we use the fact that at least $N/2$ quantum queries are needed to compute the parity of $N$ bits~\cite{BBC+01,FGGS98}.

First consider a Hamiltonian $H_1$ whose graph is a path with $N+1$ vertices. (Here the \emph{graph of $H$} has a vertex for each basis state and an edge between two vertices if the corresponding entry of $H$ is nonzero.) The Hamiltonian acts on vectors $\ket{i}$ with $i\in \{0,\ldots,N\}$ and has nonzero matrix elements
\begin{align}
\bra{i-1}H_1\ket{i}=\bra{i}H_1\ket{i-1}=\sqrt{i(N-i+1)}
\end{align}
for $i \in [N]$. Simulating $H_1$ for time $\pi/2$ starting with the state $\ket{0}$ gives the state $\ket{N}$ (i.e., $e^{-iH_1\pi/2}\ket{0}=\ket{N}$).

Next, consider a Hamiltonian $H_2$ generated from a string $x \in \{0,1\}^N$ as in \cite{BACS07}. $H_2$ acts on vertices $\ket{i,j}$ with $i\in \{0,\ldots,N\}$ and $j\in \{0,1\}$ and has nonzero matrix elements
\begin{align}
\bra{i-1,j}H_2\ket{i,j\oplus x_i}
 =\bra{i,j\oplus x_i}H_2\ket{i-1,j}
 =\sqrt{i(N-i+1)}
\end{align}
for all $i \in [N]$ and $j \in \{0,1\}$. By construction, $\ket{0,0}$ is connected to $\ket{i,j}$ if and only if $j=x_1 \oplus \cdots \oplus x_i$. In particular, $\ket{0,0}$ is connected to $\ket{N,x_1 \oplus \cdots \oplus x_N}$, and determining whether it is connected to $\ket{N,0}$ or $\ket{N,1}$ determines the parity of $x$. The graph of $H_2$ consists of two disjoint paths, one containing $\ket{0,0}$ and $\ket{N,x_1 \oplus \cdots \oplus x_N}$. Thus we have $e^{-iH_2\pi/2}\ket{0,0} = \ket{N,x_1 \oplus \cdots \oplus x_N}$, so evolution for time $\pi/2$ computes the parity of $x$.

Finally, we construct the Hamiltonian $H$ claimed in the lemma. As before, $H$ is generated from a string $x \in \{0,1\}^N$. $H$ acts on vertices $\ket{i,j,\ell}$ with $i \in \{0,\ldots,N\}$, $j \in \{0,1\}$, and $\ell \in [d]$.
The nonzero entries of $H$ are
\begin{align}
\bra{i-1,j,\ell}H\ket{i,j \oplus x_i,\ell'}
=\bra{i,j \oplus x_i,\ell'}H\ket{i-1,j,\ell}
=\sqrt{i(N-i+1)}/N
\end{align}
for all $i \in [N]$, $j \in \{0,1\}$, and $\ell,\ell' \in [d]$. The graph of $H$ is similar to that of $H_2$, except that for each vertex in $H_2$, there are now $d$ copies of it in $H$. Each vertex is connected to all $d$ copies of its neighboring vertices, so the graph has maximum degree $2d$.  Observe that, having divided the matrix elements by $N$, we have $\|H\|_{\max}=\Theta(1)$.

Now we simulate the Hamiltonian starting from the state $\ket{0,0,*}$, where $\ket{i,j,*} := \frac{1}{\sqrt{d}} \sum_\ell \ket{i,j,\ell}$ denotes a uniform superposition over the third register. The subspace $\spn\{\ket{i,j,*}: i \in \{0,\ldots,N\},$ $j \in \{0,1\}\}$ is an invariant subspace of $H$.  Since the initial state lies in this subspace, the quantum walk remains in this subspace.  
The nonzero matrix elements of $H$ in this invariant subspace are
\begin{align}
   \bra{i-1,j,*}H\ket{i,j \oplus x_i,*}
  =\bra{i,j \oplus x_i,*}H\ket{i-1,j,*} &
  =d\sqrt{i(N-i+1)}/N,
\end{align}
so we have $e^{-iHt}\ket{0,0,*} = \ket{N,x_1\oplus \cdots \oplus x_N,*}$ for $t=N\pi/2d$. Since this determines the parity of $x$, we find a lower bound of $\Omega(N) = \Omega(td)$ as claimed.
\end{proof}

It is now straightforward to use this result to prove \thm{lower}.

\begin{proof}[Proof of \thm{lower}]
We choose one of two Hamiltonians depending on whether the first or second term in \eq{lower} is larger.
If $\tau$ is larger, then we use \lem{low}.
The value of $d$ used in \lem{low} is denoted $d'$ here, to distinguish it from the $d$ given in \thm{lower}.
Taking $d'=\lfloor d/2 \rfloor$, we ensure that $d'$ is a positive integer, because $d\ge 2$.
Then \lem{low} shows that there is a $2d'$-sparse Hamiltonian; given this value of $d'$, this Hamiltonian is also $d$-sparse, as required for \thm{lower}.

For \thm{lower}, we are also given a required value for $\|H\|_{\max}$.
The Hamiltonian used in \lem{low} has $\|H\|_{\max}=\Theta(1)$.
By multiplying that Hamiltonian by a scaling factor, we obtain a Hamiltonian with the required value of $\|H\|_{\max}$.
Dividing the time used in \lem{low} by the same factor, the simulation requires time $\Omega(\tau)$ for constant precision.
In \thm{lower} we require precision $\epsilon$, which can only increase the complexity.

In the case where the second term is larger, we use Theorem 6.1 of \cite{BCCKS14}.
There it is shown that performing a simulation of a $2$-sparse Hamiltonian with precision $\epsilon$ and $\|H\|_{\max}t=O(1)$ requires
\begin{equation}
\label{eq:bccks}
\Omega \left( \frac{\log(1/\epsilon)}{\log\log(1/\epsilon)}\right)
\end{equation}
queries.
Because $d\ge2$, this Hamiltonian is also $d$-sparse.
As using larger values of $\|H\|_{\max}t$ can only increase the complexity, we also have this lower bound in the more general case.
Therefore, regardless of whether the first or second term in \eq{lower} is larger, this expression provides a lower bound on the complexity.
\end{proof}

It is also possible to combine our lower bound with the lower bound of \cite{BCCKS14} to obtain a combined lower bound in terms of $d$, $t$, and $\epsilon$, that is stronger than \thm{lower}. This yields a lower bound of $\Omega(N)$ for any $N$ that satisfies $\epsilon < \frac{1}{2} |{\sin(td/N)}|^N$. Note that when $\epsilon$ is a constant, we recover \lem{low} and when $td$ is constant, we recover \eq{bccks}. However, for intermediate values this lower bound can be strictly larger than the expression in \thm{lower}.

\section{Conclusion}
\label{sec:disc}
Our technique for Hamiltonian simulation combines ideas from quantum walks and
fractional-query simulation to provide improved performance over both previous
techniques. As a result, it provides near-optimal scaling with respect to all parameters of interest. In particular, the scaling is only slightly superlinear in $\tau = d\norm{H}_{\max}t$, whereas we have proven that linear scaling is optimal. Furthermore, the method has query complexity sublogarithmic in the allowed error, which was proven to be optimal in \cite{BCCKS14}.

Nevertheless, there is still a gap between the complexity of our algorithm and the lower bound in \eq{lower}, as they involve different tradeoffs between the parameters $\tau$ and $\epsilon$.  It remains open whether the performance can be further improved, perhaps to give performance similar to \eq{lower}, although as observed at the end of \sec{lower}, we can rule out scaling strictly as in \eq{lower}. 

Our technique can potentially be used for the more general task of operation conversion, in which we use one quantum operation to implement another.  In our work, we convert a step of a quantum walk to Hamiltonian evolution, whereas in \cite{HHL09} the task is to convert Hamiltonian evolution to an inverse. One approach to operation conversion is to use phase estimation. Here we have shown that a superposition of operations can provide far better performance.

\section*{Acknowledgment}

D.W.B. is funded by an ARC Future Fellowship (FT100100761).  This work was also supported in part by CIFAR, NSERC, the Ontario Ministry of Research and Innovation, and the US ARO under ARO grant Contract Numbers W911NF-12-1-0482 and W911NF-12-1-0486.
This preprint is MIT-CTP \#4631.

\bibliographystyle{myhamsplain}
\bibliography{sim}

\appendix
\section{Proofs of technical lemmas}
\label{app:proofs}

We now present proofs of some of the more technical results.

\begin{proof}[Proof of \lem{diamond}]
Consider two operators $U$ and $V$ acting on a pure state $\ket{\psi}$.
For the following analysis we define
\begin{align}
\phi &:= \arg{\bra{\psi}V^\dagger U\ket{\psi}}, \\
N_U &:= 1/\sqrt{\bra{\psi}U^\dagger U\ket{\psi}}, \\
N_V&:= 1/\sqrt{\bra{\psi}V^\dagger V\ket{\psi}},\\
N_{\pm} &:= \sqrt{2 \pm 2N_U N_V|\bra{\psi}V^\dagger U\ket{\psi}|}.
\end{align}
Then we define a basis
\begin{equation}
\ket{\chi_\pm} := \frac{N_U U\ket{\psi} \pm e^{i\phi}N_V V\ket{\psi}}{N_\pm} .
\end{equation}
In terms of this basis, we have
\begin{align}
U\ket{\psi} &= \frac 1{2N_U} (N_+\ket{\chi_+}  + N_-\ket{\chi_-}), \\
V\ket{\psi} &= \frac{e^{-i\phi}}{2N_V}  (N_+\ket{\chi_+}  - N_-\ket{\chi_-}),
\end{align}
so
\begin{align}
U\ket{\psi}\bra{\psi}U^\dagger - V\ket{\psi}\bra{\psi}V^\dagger &=
\frac 1{4N_U^2} \begin{bmatrix}
N_+^2 & N_+N_-  \\
N_+N_- & N_-^2 \end{bmatrix} - \frac 1{4N_V^2}
\begin{bmatrix}
N_+^2 & -N_+N_-  \\
-N_+N_- & N_-^2 \end{bmatrix} .
\end{align}

The eigenvalues of this matrix are
\begin{align}
\lambda_\pm = \frac 12 \left[\bra{\psi}U^\dagger U\ket{\psi} - \bra{\psi}V^\dagger V\ket{\psi} \pm \sqrt{(\bra{\psi}U^\dagger U\ket{\psi} + \bra{\psi}V^\dagger V\ket{\psi})^2-4|\bra{\psi}V^\dagger U\ket{\psi}|^2} \right] .
\end{align}
The square root in this expression must be at least as large as $|\bra{\psi}U^\dagger U\ket{\psi} - \bra{\psi}V^\dagger V\ket{\psi}|$, so $\lambda_-\le 0$ and $\lambda_+\ge 0$.
Thus the trace norm (i.e., the Schatten $1$-norm, denoted $\norm{\cdot}_1$) is
\begin{align}
\|U\ket{\psi}\bra{\psi}U^\dagger - V\ket{\psi}\bra{\psi}V^\dagger\|_1 &= |\lambda_+| + |\lambda_-| \nn
&=\sqrt{(\bra{\psi}U^\dagger U\ket{\psi} + \bra{\psi}V^\dagger V\ket{\psi})^2-4|\bra{\psi}V^\dagger U\ket{\psi}|^2}\, .
\end{align}
Next, from the definition of the spectral norm,
\begin{align}
\|U-V\|^2 &\ge \bra{\psi} (U-V)^\dagger (U-V) \ket{\psi} \nn
&= \bra{\psi}U^\dagger U\ket{\psi}+\bra{\psi}V^\dagger V\ket{\psi} - \bra{\psi}V^\dagger U\ket{\psi} - \bra{\psi}U^\dagger V\ket{\psi} \nn
&\ge \bra{\psi}U^\dagger U\ket{\psi}+\bra{\psi}V^\dagger V\ket{\psi} - 2|\bra{\psi}V^\dagger U\ket{\psi}|.
\end{align}
Using this inequality with the expression for the trace norm gives
\begin{align}
&\|U\ket{\psi}\bra{\psi}U^\dagger - V\ket{\psi}\bra{\psi}V^\dagger\|_1^2 \nn
&= (\bra{\psi}U^\dagger U\ket{\psi}+\bra{\psi}V^\dagger V\ket{\psi} -2|\bra{\psi}V^\dagger U\ket{\psi}|)(\bra{\psi}U^\dagger U\ket{\psi}+\bra{\psi}V^\dagger V\ket{\psi} +2|\bra{\psi}V^\dagger U\ket{\psi}|) \nn
&\le \|U-V\|^2 (\bra{\psi}U^\dagger U\ket{\psi}+\bra{\psi}V^\dagger V\ket{\psi} +2|\bra{\psi}V^\dagger U\ket{\psi}|).
\end{align}
Provided $\|U\|\le 1$ and $\|V\|\le 1$, this yields
\begin{align}
\|U\ket{\psi}\bra{\psi}U^\dagger - V\ket{\psi}\bra{\psi}V^\dagger\|_1^2 \le 4\|U-V\|^2,
\end{align}
so
\begin{align}
\|U\ket{\psi}\bra{\psi}U^\dagger - V\ket{\psi}\bra{\psi}V^\dagger\|_1 \le 2\|U-V\|.
\end{align}

Given a mixed state $\rho=\sum_j p_j\ket{\psi_j}\bra{\psi_j}$, strong convexity implies that
\begin{align}
\|U\rho U^\dagger - V\rho V^\dagger\|_1 &\le \sum_j p_j \| U\ket{\psi_j}\bra{\psi_j}U^\dagger - V\ket{\psi_j}\bra{\psi_j}V^\dagger \|_1 
\nn &\le 2\sum_j p_j \| U - V \| \nn &= 2 \| U - V \|  .
\end{align}
Similarly, tensoring with the identity gives
\begin{align}
\|(U\otimes\openone)\rho (U^\dagger\otimes\openone) - (V\otimes\openone)\rho (V^\dagger\otimes\openone)\|_1 &\le 2 \| U\otimes\openone - V\otimes\openone \| \nn
&= 2 \| U - V \|.
\end{align}
Hence, maximizing over $\rho$, the diamond norm satisfies
\begin{equation}
\| U-V \|_\diamond \le 2 \| U - V \|.
\end{equation}

Now consider the case where $U$ is some desired unitary, and $V=V_0$ is an operation that happens if a measurement succeeds, with other operations $V_j$ occurring for other measurement outcomes.
That is, the overall channel is
\begin{equation}
C(\rho) = V\rho V^\dagger + \sum_{j} V_j \rho V_j^\dagger
\end{equation}
where $V^\dag V + \sum_j V_j^\dag V_j = I$.
The trace is bounded by
\begin{align}
{\rm Tr}(V\rho V^\dagger) &={\rm Tr}(U\rho U^\dagger) - {\rm Tr}(U\rho U^\dagger-V\rho V^\dagger) \nn
&\ge 1 - \| U\rho U^\dagger - V\rho V^\dagger \|_1 \nn
&\ge 1- 2 \| U - V \|.
\end{align}
Hence the trace of the remaining part is bounded as
\begin{align}
{\rm Tr}\left(\sum_{j} V_j \rho V_j^\dagger\right) &= {\rm Tr}(C(\rho)) - {\rm Tr}(V\rho V^\dagger) \nn
&\le 1 - (1-2 \| U - V \|) \nn
&\le 2 \| U - V \|.
\end{align}
For non-negative Hermitian operators, the trace norm is equal to the trace, so
\begin{equation}
\left\|\sum_{j} V_j \rho V_j^\dagger\right\|_1 \le 2 \| U - V \|.
\end{equation}
Since the trace is unchanged by tensoring with the identity, we have
\begin{equation}
\left\|\sum_{j} (V_j\otimes\openone) \rho (V_j\otimes\openone)^\dagger\right\|_1 \le 2 \| U - V \|.
\end{equation}
Hence
\begin{align}
\| C-U\|_\diamond &= \max_\rho \left\| (V\otimes\openone)\rho (V\otimes\openone)^\dagger + \sum_{j} (V_j\otimes\openone) \rho (V_j\otimes\openone)^\dagger - (U\otimes\openone)\rho (U\otimes\openone)^\dagger \right\|_1 \nn
&\le \max_\rho \left[ \| (V\otimes\openone)\rho (V\otimes\openone)^\dagger - (U\otimes\openone)\rho (U\otimes\openone)^\dagger \|_1 + \left\| \sum_{j} (V_j\otimes\openone) \rho (V_j\otimes\openone)^\dagger \right\|_1 \right]\nn
&\le 4 \|U-V\|
\end{align}
as claimed.
\end{proof}

\begin{proof}[Proof of \lem{errorbound}]
For $|m|\le k$, the values of $a_m$ differ from $J_m(\seg)$ only due to the normalization factor in \eq{avals}.
The Bessel function $J_m(\seg)$ is bounded, for real $z$ and integer $m$, by \cite[9.1.62]{AS64}
\begin{equation}
\label{eq:jsim}
|{J_m(\seg)}| \le \frac{1}{|m|!} \left| \frac{\seg}{2} \right|^{|m|}
\end{equation}
(here we use the fact that $J_{-m}(\seg) = (-1)^m J_m(\seg)$ \cite[9.1.5]{AS64}).
Using this bound, together with the condition that $|z|\le k$,
\begin{align}
\label{eq:Jbnd}
2\sum_{m=k+1}^\infty |J_m(\seg)|
\le 2\sum_{m=k+1}^\infty \frac{|\seg/2|^m}{m!}
< 2\frac{|\seg/2|^{k+1}}{(k+1)!}\sum_{m=k+1}^{\infty} (1/2)^{m-(k+1)} 
= 4\frac{|\seg/2|^{k+1}}{(k+1)!}.
\end{align}
As a result, the normalization factor in \eq{avals} satisfies
\begin{equation}
\sum_{m=-k}^{k} J_m(\seg) \ge 1-4\frac{|\seg/2|^{k+1}}{(k+1)!}.
\end{equation}
This means that $a_m$ closely approximates $J_m(z)$, in the sense that
\begin{equation}
\label{eq:aJ}
a_m = J_m(z) \left[ 1+ O\left(\frac{(\seg/2)^{k+1}}{(k+1)!}\right)\right].
\end{equation}

Similarly, using $|\mu_\pm^{m}-1|\le |\nu m|$ and $|\seg|\le k$ gives
\begin{align}
\label{eq:Jbnd2}
\left|2\sum_{m=k+1}^\infty J_m(\seg) (\mu_\pm^m-1) \right|
&\le 2|\nu|\sum_{m=k+1}^\infty m\frac{|\seg/2|^m}{m!} \nn
&< 2|\nu|\frac{|\seg/2|^{k+1}}{(k+1)!}\sum_{m=k+1}^{\infty} m(1/2)^{m-(k+1)} \nn
&= 4|\nu|(k+2)\frac{|\seg/2|^{k+1}}{(k+1)!}.
\end{align}
Using \eq{lamsum} gives
\begin{equation}
\label{eq:expan}
e^{i\nu \seg}-1 = \sum_{m=-\infty}^{\infty} J_m(\seg) (\mu_\pm^{m}-1) .
\end{equation}
Therefore, with \eq{Jbnd2}, we obtain
\begin{align}
 \sum_{m=-k}^{k} J_m(\seg) (\mu_\pm^{m}-1) &= e^{i\nu \seg}-1 + O\left(\nu\frac{(\seg/2)^{k+1}}{k!}\right).
 \end{align}
Using this expression together with \eq{aJ} then gives
\begin{align}
\sum_{m=-k}^{k} a_m (\mu_\pm^m-1) &= \left[e^{i\nu \seg}-1 + O\left(\nu\frac{(\seg/2)^{k+1}}{k!}\right) \right]\left[ 1+ O\left(\frac{(\seg/2)^{k+1}}{(k+1)!}\right) \right].
\end{align}
Using $|e^{i\nu \seg}-1|\le |\nu z|$ and $|z|\le k$ gives
\begin{align}
\sum_{m=-k}^{k} a_m (\mu_\pm^m-1) &= e^{i\nu \seg}-1 + O\left(\nu\frac{(\seg/2)^{k+1}}{k!}\right).
\end{align}
Because $\sum_m a_m = 1$, we obtain
\begin{equation}
\left|\sum_{m=-k}^{k} a_m \mu_\pm^{m} - e^{i\nu \seg} \right| = 
O\left( \nu \frac{(\seg/2)^{k+1}}{k!}  \right).
\end{equation}
Now $\sum_{m=-k}^{k} a_m \mu_\pm^{m}$ are the eigenvalues of $V_k$, and $ e^{i\nu \seg}$ are the eigenvalues of the desired unitary operation $V_\infty$.
Using $\nu=\lambda/Xd$, we have $|\nu|\le\norm{H}/Xd$.
Hence the norm of the difference of these operators is bounded as in \eq{error}.
\end{proof}

\begin{proof}[Proof of \lem{sqrtsum}]
First we bound the sum over small values of $m$.  The Bessel functions satisfy \cite[9.1.76]{AS64}
\begin{equation}
\sum_{m=-\infty}^{\infty} [J_m(\seg)]^2 = 1.
\end{equation}
For any positive integer $k$ we have
\begin{equation}
\sum_{m=-k}^{k} |J_m(\seg)|^2 < 1.
\end{equation}
Using the Cauchy-Schwarz inequality, we find
\begin{equation}
\label{eq:boundsmall}
\sum_{m=-k}^k |J_m(\seg)| \le \sqrt{\sum_{m=-k}^{k} 1}\sqrt{\sum_{m=-k}^{k} |J_m(\seg)|^2} < \sqrt{2k+1}.
\end{equation}

Provided $|z|\le k$, using \eq{Jbnd}, together with $\ell ! > (\ell/e)^\ell$, we obtain
\begin{align}
\label{eq:boundbig}
2\sum_{m=k+1}^\infty |J_m(\seg)| < 4\left|\frac{e \seg}{2(k+1)}\right|^{k+1}.
\end{align}
Combining \eq{boundsmall} and \eq{boundbig}, we find
\begin{align}
  {\cal S}(\seg) < \sqrt{2k+1} + 4\left|\frac{e \seg}{2(k+1)}\right|^{k+1}.
\end{align}
Finally, taking $k=\lceil e\seg/2 \rceil$, the second term is $O(1)$, which gives ${\cal S}(\seg) = O(\sqrt{|\seg|})$ as claimed.
\end{proof}

\end{document}